\documentclass[11pt]{article}

\usepackage{amsfonts,amssymb,amsmath,latexsym,ae,aecompl}
\usepackage{graphics}
\usepackage{latexsym}
\usepackage{enumerate}
\usepackage{algorithm}
\usepackage{algorithmic}
\usepackage{tikz}
\usepackage{caption}
\usepackage{subcaption}
\usepackage{xspace}
\newcommand{\qed}{\hfill$\Box$}

\newenvironment{proof}{\noindent {\bf Proof.}}{\qed}

\newtheorem{theorem}{Theorem}[section]
\newtheorem{lemma}{Lemma}[section]

\newtheorem{corollary}{Corollary}[section]

\newtheorem{proposition}{Proposition}[section]

\begin{document}

\newcommand{\cA}{{\cal A}}
\newcommand{\cB}{{\cal B}}
\newcommand{\cC}{{\cal C}}
\newcommand{\cZ}{{\cal Z}}
\newcommand{\cG}{{\cal G}}
\newcommand{\cN}{{\cal N}}
\newcommand{\cU}{{\cal U}}
\newcommand{\cT}{{\cal T}}
\newcommand{\cS}{{\cal S}}
\newcommand{\cP}{{\cal P}}
\newcommand{\cL}{{\cal L}}
\newcommand{\cV}{{\cal V}}
\newcommand{\cH}{{\cal H}}
\newcommand{\loc}{{\cal LOCAL}}
\newcommand{\cY}{{\cal Y}}
\newcommand{\ai}{\alpha_i}
\newcommand{\bi}{\beta_i}
\newcommand{\gi}{\gamma_i}
\newcommand{\di}{\delta_i}









\bibliographystyle{plain}


\title{Finding the Size and the Diameter of a Radio Network Using Short Labels
\footnotemark[1]}

\author{Barun Gorain\footnotemark[2]
 \and
  Andrzej Pelc\footnotemark[3]}
\date{ }
\maketitle
\def\thefootnote{\fnsymbol{footnote}}
\footnotetext[1]{A preliminary version of this paper appeared in Proc.  19th International Conference on Distributed Computing and Networking
(ICDCN 2018).}
\footnotetext[2]{
\noindent
Department of Electrical Engineering and Computer Science, Indian Institute of Technology Bhilai, India.
{\tt barun@iitbhilai.ac.in}}
\footnotetext[3]{
\noindent
D\'epartement d'informatique, Universit\'e du Qu\'ebec en Outaouais, Gatineau,
Qu\'ebec J8X 3X7, Canada, {\tt pelc@uqo.ca}. Partially supported by NSERC discovery grant 2018-03899 and by the Research Chair in Distributed Computing at the
Universit\'e du Qu\'ebec en Outaouais.}

%
%
%

\begin{abstract}
The number of nodes of a network, called its {\em size}, and the largest distance between nodes of a network, called its {\em diameter}, are among the most important network parameters. Knowing the size
and/or diameter (or a good upper bound on those parameters) is a prerequisite of many distributed network algorithms, ranging from broadcasting and gossiping, through leader election, to rendezvous and exploration.
A radio network is a collection of stations, called nodes, with wireless transmission and receiving capabilities. It is modeled as a simple connected undirected graph whose
nodes communicate in synchronous rounds. In each round, a node can either transmit a message to all its neighbors, or stay silent and listen. At the receiving end, a node $v$ hears a message from a neighbor $w$ in a given round, if $v$ listens in this round, and if $w$ is its only neighbor that transmits in this round.  If $v$ listens in a round,
and two or more neighbors of $v$ transmit in this round, a {\em collision} occurs at $v$. If $v$ transmits in a round, it does not hear anything in this round. Two scenarios are considered in the literature: if listening nodes can distinguish collision
from silence (the latter occurs when no neighbor transmits), we say that the network has the {\em collision detection} capability, otherwise there is no collision detection.

We consider the tasks of {\em size discovery} and {\em diameter discovery}: finding the size (resp. the diameter) of an unknown radio network with collision detection. All nodes have to output the size
(resp. the diameter) of the network, using a deterministic algorithm.
Nodes have  labels which are (not necessarily distinct)
binary strings. The length of a labeling scheme is the largest length of a label.

We concentrate on the following problems:
\begin{quotation}
What is the shortest labeling scheme that permits size discovery in all radio networks of maximum degree $\Delta$?
What is the shortest labeling scheme that permits diameter discovery in all radio networks?
\end{quotation}

Our main result states that the minimum length of a labeling scheme that permits size discovery is $\Theta(\log\log \Delta)$. The upper bound is
proven by designing a size discovery algorithm using a labeling scheme of length $O(\log\log \Delta)$, for all networks  of maximum degree $\Delta$.
The matching lower bound is proven by constructing a class of graphs (in fact even of trees) of maximum degree $\Delta$,
for which any size discovery algorithm must use a labeling scheme of length at least $\Omega(\log \log\Delta)$ on some graph of this class.
By contrast, we show that diameter discovery can be done in all radio networks using a labeling scheme of constant length.

\end{abstract}

\section{Introduction}

\subsection{The model and the problem}

The number of nodes of a network, called its {\em size}, and the largest distance between nodes of a network, called its {\em diameter}, are among the most important network parameters. Knowing the size
and/or diameter (or a good upper bound on those parameters) by nodes of a network or by mobile agents operating in it, is a prerequisite of many distributed network algorithms, ranging from broadcasting and gossiping, through leader election, to rendezvous and exploration.

A radio network is a collection of stations, called nodes, with wireless transmission and receiving capabilities. It is modeled as a simple connected undirected graph.
As it is usually assumed in the algorithmic theory of radio networks \cite{CGR,GPPR,GPX},  all nodes start simultaneously and communicate in synchronous rounds. In each round, a node can either transmit a message to all its neighbors, or stay silent and listen. At the receiving end, a node $v$ hears a message from a neighbor $w$ in a given round, if $v$ listens in this round, and if $w$ is its only neighbor that transmits in this round. If $v$ listens in a round,
and two or more neighbors of $v$ transmit in this round, a {\em collision} occurs at $v$. If $v$ transmits in a round, it does not hear anything in this round. Two scenarios are considered in the literature: if listening nodes can distinguish collision
from silence (the latter occurs when no neighbor transmits), we say that the network has the {\em collision detection} capability, otherwise there is no collision detection.

We consider the tasks of {\em size discovery} and {\em diameter discovery}: finding the size (resp. the diameter) of an unknown radio network with collision detection. All nodes have to output the size
(resp. the diameter) of the network, using a deterministic algorithm.
Nodes have  labels which are (not necessarily distinct)
binary strings. These labels are given to (otherwise anonymous) nodes by an oracle knowing the network, whose aim is to help the nodes in executing a size or diameter discovery algorithm using these labels. Such {\em informative labeling schemes}, also referred to as {\em advice} given to nodes, have been previously studied, e.g.,  in the context of ancestor queries \cite{AKM01}, MST computation \cite{FKL}, and topology recognition \cite{FPP}, for wired networks, and in the context of topology recognition \cite{GP} and broadcasting \cite{EGMP} for radio networks.
The length of a labeling scheme is the largest length of a label. A priori, every node knows only its own label.

In this paper we concentrate on the problem of finding a shortest labeling scheme permitting size and diameter discovery in radio networks with collision detection.
Clearly, some labels have to be given to nodes, because otherwise (in anonymous radio networks) no deterministic communication is possible. Indeed, for any deterministic algorithm in an anonymous network, all nodes would transmit in exactly the same rounds, and hence no node would ever hear anything.  On the other hand, labeling schemes of
length $\Theta (\log n)$, for $n$-node networks, are certainly enough to discover the size of the network, as it can be then coded in the labels. Similarly, length $\Theta (\log D)$ is enough to discover the diameter $D$. Our aim is to answer the following questions.
\begin{quotation}
What is the shortest labeling scheme that permits size discovery in all radio networks of maximum degree $\Delta$?
What is the shortest labeling scheme that permits diameter discovery in all radio networks? \end{quotation}

\subsection{Our results}

Our main result states that the minimum length of a labeling scheme that permits size discovery is $\Theta(\log\log \Delta)$. The upper bound is
proven by designing a size discovery algorithm using a labeling scheme of length $O(\log\log \Delta)$, for all networks  of maximum degree $\Delta$.
The matching lower bound is proven by constructing a class of graphs (in fact even of trees) of maximum degree $\Delta$,
for which any size discovery algorithm must use a labeling scheme of length at least $\Omega(\log \log\Delta)$ on some graph of this class.
By contrast, we show that diameter discovery can be done in all radio networks using a labeling scheme of constant length.

\subsection{Related work}

Algorithmic problems in radio networks modeled as graphs were studied for such distributed tasks as broadcasting \cite{CGR,GPX}, gossiping \cite{CGR,GPPR} and leader election
\cite{CD,KP}. In some cases \cite{CGR,GPPR}, the model without collision detection was used, in others \cite{GHK,KP}, the collision detection capability was assumed.

Providing nodes of a network, or mobile agents circulating in it, with information of arbitrary type (in the form of binary strings) that can be used by an algorithm to perform some network task, has been
proposed in \cite{AKM01,CFIKP,DP,EFKR,FGIP,FIP1,FIP2,FKL,FP,GPPR02,IKP,KKKP02,KKP05,SN}. This approach was referred to as
algorithms using {\em informative labeling schemes}, or equivalently, algorithms with {\em advice}.
When advice is given to nodes,  two variations are considered: either the binary string given to nodes is the same for all of them \cite{GMP} or different strings may be given to different nodes
\cite{FKL,FPP}, as in our present case.
 If strings may be different, they can be considered as labels assigned to (otherwise anonymous) nodes.
Several authors studied the minimum length of labels required for a given network problem to be solvable, or to solve a
network problem in an efficient way. The framework of advice or of labeling schemes permits us to quantify the amount of needed information,
regardless of the type of information that is provided and of the way the algorithm subsequently uses it.

In \cite{FIP1}, the authors compared the minimum size of advice required to
solve two information dissemination problems, using a linear number of messages.
 In \cite{KKP05}, given a distributed representation of a solution for a problem,
the authors investigated the number of bits of communication needed to verify the legality of the represented solution.
 In \cite{FIP2}, the authors
established the size of advice needed to break competitive ratio 2 of an exploration algorithm in trees.
In \cite{FKL}, it was shown that advice of constant size permits to carry out the distributed construction of a minimum
spanning tree in logarithmic time.
In \cite{GPPR}, short labeling schemes were constructed with the aim to answer queries about the distance between any pair of nodes.
In \cite{EFKR}, the advice paradigm was used for online problems.
In the case of \cite{SN}, the issue was not efficiency but feasibility: it
was shown that $\Theta(n\log n)$ is the minimum size of advice
required to perform monotone connected graph clearing.

There are three papers studying the size of advice in the context of radio networks.
In \cite{IKP}, the authors studied radio networks without collision detection for
which it is possible to perform centralized broadcasting in constant time. They proved that
a total of $O(n)$ bits of additional information (i.e., not counting the labels of nodes) given to all nodes are sufficient
for performing broadcast in constant time in such networks, and a total of
$o(n)$ bits are not enough. In \cite{GP}, the authors considered the problem of topology recognition in wireless trees without collision detection.
Similarly to the present paper, they investigated short labeling schemes permitting to accomplish this task.
It should be noted that the results in \cite{GP} and in the present paper are not comparable:
\cite{GP} studies a harder task (topology recognition) in a weaker model (no collision detection), but restricts attention only to trees, while the present paper studies easier tasks (size and diameter discovery) in a stronger model (with collision detection) but our results hold for arbitrary networks.
In a recent paper \cite{EGMP}, the authors considered the problem of broadcasting in radio networks without collision detection, and proved that this can be done using a labeling scheme of constant length.

\section{Preliminaries}

According to the definition of labeling schemes, a label of any node should be a finite binary string.
For ease of comprehension, in our positive result concerning size discovery, we present our labels in a more structured way, namely as sequences $(a,b,c,d)$, where $a$ is a binary string of length 7, and each of $b$, $c$ and $d$ is a pair whose first term is a binary string, and the second term is a bit.
Each of the components $a$, $b$, $c$, $d$, is later used in the size discovery algorithm in a particular way. It is well known that such a sequence $(a,b,c,d)$ can be unambiguously coded as a single binary string whose length is a constant multiple of the sum of lengths of all binary strings that compose it. Hence, presenting labels in this more structured way and skipping the details of the encoding does not change the order of magnitude
of the length of the constructed labeling schemes.

In our algorithms, we use the subroutine $Wave(x)$, for a positive integer $x$, that can be implemented in radio networks with collision detection
(cf. \cite{CGGPR} where a similar procedure was called Algorithm {\tt Encoded-Broadcast}). We describe the subroutine below,  for the sake of completeness. The aim of $Wave(x)$ is to transmit the integer $x$ to all nodes of the network, bit by bit.
During the execution of $Wave(x)$, each node is colored either blue, or red or white. Blue and white nodes know $x$ and after each phase red neighbors of blue nodes learn $x$ and become blue, while blue nodes become white. $Wave(x)$ is initiated by some node $v$. At the beginning, $v$ is blue and all other nodes are red.

Let $p=(a_1a_2\dots a_k)$ be the binary representation of the integer $x$. Consider the binary sequence $p^*=(b_1,b_2, \dots, b_{2k+2})$ of length $2k+2$ that is formed from $p$ by replacing every bit 1 by 10, every bit 0 by 00, and adding 11 at the end. For example, if $p=(1101)$ then $p^*=(1010001011)$.  Each phase of $Wave(x)$ lasts $2k+2$ rounds, starting in some round $r+1$.
In consecutive rounds $r+1,\dots ,r+2k+2$, every blue node transmits some message $m$ in round $r+i$, if $b_i=1$, and remains silent if $b_i=0$. A red node $w$ listens until a round when it hears either a collision or a message (this is round $r+1$), and then until two consecutive rounds occur when it hears either a collision or a message. Suppose that the second of these two rounds is round $s$.  Then $w$ decodes $p^*$  by putting $b_i=1$ if it heard a message or a collision in round $t+i$, and putting $b_i=0$ if it heard silence in round $t+i$. From $p^*$ it computes unambiguously $p$ and then $x$. Round $s=r+2k+2$ is the round in which all blue nodes finished transmitting in the current phase of $Wave(x)$. In round $s+1$ which starts the next phase, all blue nodes become white and all red nodes that heard a collision or a message in round $s$ become blue.

  In this way, the subroutine $Wave(x)$ proceeds from level to level, where the $i$-th level is formed by the nodes at distance $i$ from $v$ in the graph.
  Every node at a level $i>0$, is involved in the subroutine in two consecutive phases, first as a red node and then as a blue node. The initiating node $v$ is involved only in the first phase.
  Since a sequence of the form $p^*$ cannot be a prefix of another sequence of the form $q^*$, every node can determine when the transmissions from the previous level are finished, and can correctly decode $x$. In our applications, no other transmissions are performed simultaneously with transmissions prescribed by $Wave(x)$,
  and hence nodes can compute when a given $Wave$ will terminate.

\section{Finding the size of a network}

This section is devoted to the task of finding the size of a network.

\subsection{The Algorithm {\tt Size Discovery}}

In this section, we construct a labeling scheme of length $O(\log\log \Delta)$ and a size discovery algorithm using this scheme and working for any radio network of maximum degree $\Delta$.

\subsubsection{Construction of the labeling scheme}

Let $G$ be a graph of maximum degree $\Delta$. Let $r$ be any node of $G$ of degree $\Delta$. For $l \ge 0$, a node is said to be in level $l$, if its distance from $r$ is $l$. Let $h$ be the maximum level in $G$. Let $V(l)$ be the set of nodes in level $l$.
For any node $v\in V(l)$, let $N(v)$ be the set of neighbors of $v$ which are in level $l+1$.

Before giving the detailed description of the labeling scheme and of the algorithm, we give a high-level idea of our size discovery method.
The algorithm is executed level by level in a bottom up fashion. Each node of a level maintains an integer variable, $weight$, such that the sum of the weights of all nodes in level $l$, for $0\le l\le h$, is equal to the total number of nodes in levels $l'\ge l$.  Using the assigned labels, these weights are transmitted to a special set of nodes, called {\em upper set}, in level $l-1$. An upper set in level $l-1$ is an ordered subset of the nodes in level $l-1$, which covers all the nodes in level $l$, i.e., each node in level $l$ is a neighbor of at least one node in the upper set of level $l$. Using this property of upper sets, the capability of collision detection, and a specially designed labeling of the nodes, multiple accounting of the weights is prevented.
The weights of each level are transmitted up the tree, and finally the node at level 0, i.e., the root calculates its weight, which is the size of the network. In the final stage, this size is transmitted to all other nodes.

%

Let $U=\{v_1,v_2,\cdots ,v_k\}$ be an ordered set of nodes in level $l$. For all $v \in V(l)\setminus U$, we define $N'(v,U)= N(v)\setminus (\cup_{w \in U}N(w))$ and for $v_j\in U$, $N'(v_j,U)=N(v_j)\setminus (\cup_{i=1}^{j-1}N(v_i))$. An ordered subset $U$ of $V(l)$ is said to be an {\em upper set} at level $l$, if for each $v \in U$, $N'(v,U) \ne \emptyset$ and $\cup_{w\in U}N(w)=V(l+1)$.

{Fig. \ref{fig:fig1} shows an example of an upper set $U=\{v_1,v_3,v_4\}$ for a two-level graph. The node $v_2$ is not part of the upper set, as the set of neighbors of $v_2$ is a subset of the neighbors of $v_1$.}

\begin{figure}[h]
\centering
\includegraphics[width=0.5\textwidth]{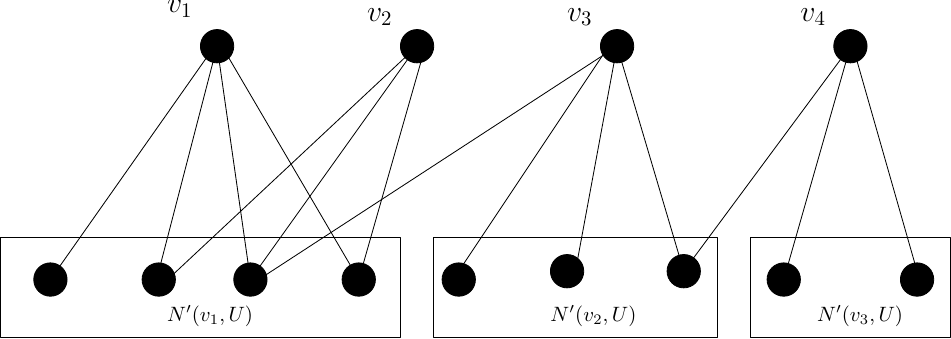}
\caption{Example of an upper set $U=\{v_1,v_3,v_4\}$}
\label{fig:fig1}
\end{figure}

Below, we propose an algorithm that computes an upper set at each level $l$, for $1\le l\le h-1$.
The algorithm works in a recursive way. The first node $v_1$ of the set is chosen arbitrarily. At any step, let $US(l)=\{v_1, \cdots, v_i\}$ be the set computed by the algorithm in the previous step. Let $u_1^j,u_2^j, \cdots, u_{|N'(v_j,US(l))|}^j$ be nodes in $N'(v_j,US(l))$ for $1 \le j\le i$. Let $k_j=1+\log \lfloor|N'(v_j,US(l))|\rfloor$. If $V(l+1)\setminus (\cup_{w\in US(l)}N(w)) \ne \emptyset$, then the next node in $US(l)$ is added using the following rules.

\begin{enumerate}
\item  Find the last node $v_a$ in $US(l)$ that has a common neighbor in $\{u_1^a,u_2^a, \cdots u_{k_a}^a\}$ with some node  $v \in V(l)\setminus US(l)$ such that $N(v)\setminus (\cup_{w\in US(l)}N(w)) \ne \emptyset$.  Choose such a node $v$ in $US(l)$ that has the common neighbor $u_b^a$ with $v_a$, where $b=\min\{1,2,\cdots,k_a\}$. Add $v$ to $US(l)$ as the node $v_{i+1}$.
\item If no such node in $v_a$ exists in $US(l)$, add any node $v \in V(l)\setminus US(l)$ with $N(v)\setminus (\cup_{w\in US(l)}N(w)) \ne \emptyset$ as the node $v_{i+1}$.
\end{enumerate}
The construction of $US(l)$ is completed when $\cup_{w\in US(l)}N(w)=V(l+1)$.

Also, for every node $v_i \in US(l)$, the nodes $u_1^i,u_2^i, \cdots, u_{k_i}^j$ are assigned some unique id's from the set $\{1,2,\cdots, \lfloor \log \Delta\rfloor+1\}$. Moreover,  if a node $v_m$ is added to $US(l)$ according to the first rule, where $v_m$ has a common neighbor $u_c^i$ with $v_i$, then the node $u_1^m$ gets the same id as $u_c^i$. If a node $v_m$ is added according to the 2nd rule, then $u_1^m$ gets the id 1.  These id's will be later used to construct the labels of  the nodes. In Algorithm \ref{alg:alg1} we give the pseudocode of the procedure that constructs
an upper set $US(l)$ for each level $l$, and that assigns id's to some nodes of $V(l+1)$, as explained above. This procedure uses in turn the subroutine  \textsc{Compute$(v,j)$} whose pseudocode is presented in Algorithm \ref{alg:2}..

\begin{algorithm}
\caption{\textsc{$ComputeSet(l)$}}
\begin{algorithmic}[1]
\STATE{$US(l)\leftarrow\{\}$, $count\leftarrow 0$}
\STATE{$V'(l) \leftarrow V(l)$ }
\FOR{all $v\in V(l)$}
\STATE{$N'(v,US(l))\leftarrow N(v)$}
\ENDFOR
\FOR{ all $v \in V'(l)$}
\STATE{ \textsc{Compute$(v,1)$}}
\ENDFOR
\STATE{Return $US(l)$}
\end{algorithmic}\label{alg:alg1}
\end{algorithm}

\begin{algorithm}
\caption{\textsc{Compute$(v,j)$}}
\begin{algorithmic}[1]
\STATE{$count\leftarrow count+1$}
\STATE{$v_{count}\leftarrow v$}
\STATE{$ID(v)\leftarrow \{j\}$}
\STATE{$US(l)\leftarrow US(l) \cup \{v_{count}\}$}
\FOR{all nodes $v'\in V'(l)\setminus US(l)$}
\STATE{$N'(v',US(l)) \leftarrow N'(v',US(l)) \setminus (\cup_{w\in US(l)}N(w))$}
\IF{$N'(v') =\emptyset$}
 \STATE{$V'(l) \leftarrow V'(l) \setminus \{v'\}$}
\ENDIF
\ENDFOR
\IF{$V'(l) \ne \emptyset$}
\STATE{Let $N'(v,US(l))=\{u_1,u_2,\cdots,u_k\}$}
\FOR{$i=1$ to $\lfloor \log k \rfloor$+1}
\label{st:stp2}\IF{$i=1$}
\STATE{$p \leftarrow j$}
\ELSE
\STATE{$p \leftarrow \min\left(\{1,2,\cdots,\lfloor\log \Delta\rfloor+1\}\setminus ID(v)\right)$}
\ENDIF
\STATE{$id(u_i) \leftarrow p$}
\label{st:stp1}\WHILE{there exists some $v'' \in V'(l)$ such that $u_i \in N(v'') $}
\STATE{\textsc{Compute$(v,p)$}}
\ENDWHILE
\STATE{$ID(v) \leftarrow ID(v) \cup \{p\}$}
\ENDFOR
\ENDIF
\end{algorithmic}\label{alg:2}
\end{algorithm}

The nodes in the upper set $US(l)$ and the assignment of their ids are shown in Fig. \ref{fig:example}. The node $v_1$ is chosen in $US(l)$ arbitrarily. The number of neighbors of $v_1$ in $V(l+1)$ is 5 (the node $v_1$ and its neighbors in $V(l+1)$ are shown as the gray circles.). Three neighbors of $v_1$, (as $1+ \lfloor \log 5 \rfloor =3$) are assigned ids 1,2,3 as shown in the figure. Note that the node with id 3 is also a neighbor of another node in $V(l)$. Hence, this node is selected as the next node in $US(l)$, according to Step \ref{st:stp1} of Algorithm \ref{alg:2}.  $N'(v_2, US(l))$ is 4 (the node $v_2$ and the nodes in $N'(v_2, US(l))$ are shown as the dotted pink circles). According to Step \ref{st:stp2} of Algorithm \ref{alg:2}, the neighbor $u_2^1$ of $v_2$ gets the same id 3 as the node $u_1^3$. In a similar fashion, the node $v_3$ (shown as dashed green circle) is chosen as the next node in $US(l)$. After adding $v_3$, no node can be added in $US(l)$ according to the first rule. The node $v_4$ is chosen according to the second rule and added to $US(l)$ and the construction of $US(l)$ is complete, as all nodes of $V(l+1)$ are taken care of.

\begin{figure}[h]
\centering
\includegraphics[width=1.0\textwidth]{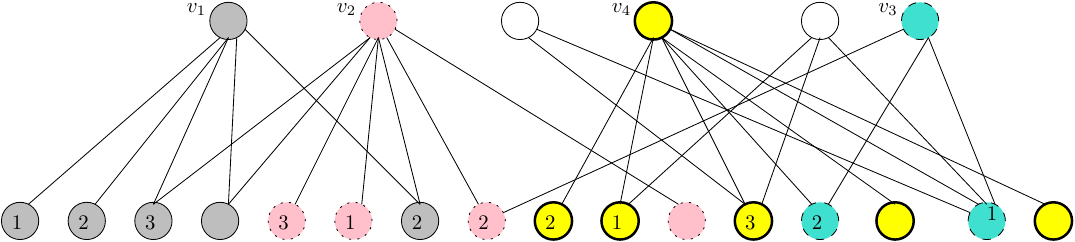}
\caption{Showing  the assignment of Ids of the nodes in $V(l+1)$}
\label{fig:example}
\end{figure}

For $0\le l\le h$, we define the {\em weight} $W(v)$ of a node $v \in V(l)$ as follows. If $v \in V(h)$, we define $W(v)=1$. For a node $v \in V(l)$, where $0\le l\le h-1$, we define $W(v)=1+\sum_{u\in N'(v,US(l))} W(u)$. Thus, for any level $l$, the sum of the weights of nodes at level $l$ is equal to the total number of nodes in levels $l'\geq l$.
Hence the weight of the node $r$ is the size of the network.

The ids assigned to the nodes in the above fashion have the following purpose. The objective of the algorithm is to let the nodes in $US(l)$  learn jointly the sum of the weights of the nodes in levels $ \ge l$. This can only be done if the weight of a node in level $l+1$ is transmitted to exactly one node in level $l$. The ids assigned to the nodes in the level $l+1$ are helpful in order to ensure this. In Fig. \ref{fig:example}, the node $v_2$ has two neighbors in level $l+1$ with id 3. The algorithm asks every node with a specific id to transmit in a specific round in different phases. Now, when the nodes with id 3 transmit, the node $v_1$ successfully receives the messages, but a collision happens at $v_2$. The node $v_2$ immediately learns that the ongoing message transmissions are dedicated to some other node and therefore it ignores the activities for the remaining rounds in the current phase.  Once the node $v_1$ learns its weight, it asks its neighbors in $N'(v_1,US(l))$ not to participate in the subsequent rounds. In the next phase, $v_2$ does not hear any collision (as one node with id 3 does not transmit) while its neighbors with positive ids transmits, hence it successfully learns its weight. But the node $v_3$ hears a collision when the node with id 2 transmits. Therefore, $v_3$ ignores all the activities in this phase. In the next phase, $v_3$ successfully learns its weight. In this way the consecutive nodes which are added in $US(l)$ by rule 1 learn their weights one by one in different phases, hence no multiple accounting can occur.

We are now ready to define the labeling scheme $\Lambda$ that will be used by our size discovery algorithm. The label $\Lambda(v)$ of each node $v$ contains two parts. The first part is a vector of markers that is a binary string of length 7, used to identify nodes with different properties.
The second part is a vector of three {\em tags}. Each tag is a pair $(id,b)$, where $id$ is the binary representation of  an integer from the set $\{1,2,\cdots, \lfloor\log \Delta\rfloor+1\}$, and $b$ is either 0 or 1. Every node will use the tags to identify the time slot when it should transmit  and what it should transmit in this particular time slot.

We first describe how the markers are assigned to different nodes of $G$.
\begin{enumerate}
\item The node $r$ gets the marker 0, and one of the nodes in level $h$ gets the marker 1.
\item Choose any set of $\lfloor\log \Delta \rfloor+1$ nodes in $N(r)$ and give them the marker 2.
\item Let $P$ be a simple path from $r$ to the node with marker 1. All the internal nodes in $P$ get the marker 3.
\item For each $l$,  $0\le l\le h-1$, all the nodes in $US(l)$ get the marker 4. The last node of $US(l)$ gets the marker 5 and a unique node from $V(l+1)$ with maximum weight in this set gets the marker 6.
\end{enumerate}

The first part of every label is a binary string  $M$ of length 7, where the markers are stored. Note that a node can be marked by multiple markers. If the node is marked by the marker $i$, for $i=0,\dots ,6$, we have $M(i)=1$; otherwise, $M(i)=0$.

The markers are assigned to the nodes in the network in order to identify different types of nodes that play different roles in the proposed algorithm. Some specific rounds are allotted to  each level during which all the nodes of that level transmit their weights. Every node learns from its label in which time slot it has to transmit.
The root is distinguished by the marker 0. In the algorithm, the node with marker 0 first learns the value of its degree which is $\Delta$, using the messages transmitted by the nodes with marker 2.  Then it transmits this value to all other nodes using Subroutine $Wave(\Delta)$.
The node with marker 1 is recognized as one of the nodes at the last level in the BFS tree rooted at the node $r$. This node is the first that learns the value of $h$ and then transmits it to $r$ using the internal  nodes, which are marked by marker 3, in the shortest path from this node to $r$. Markers 4 and 5 are used to identify nodes which are responsible for transmitting the value of the weights, and the node assigned marker 6 is the node in a level which transmits its weight last among the nodes in that level.

The second part of the label of each node $v$ is a vector $[L_1(v),L_2(v), L_3(v)]$ containing three  tags, namely, the {\it $\Delta$-learning tag} $L_1(v)$, the {\it collision tag} $L_2(v)$, and the {\em weight-transmission tag} $L_3(v)$. The assignment of the above tags is described below.

\begin{enumerate}
\item The $\Delta$-learning tags will be used for learning the value of $\Delta$ by the root $r$. The node $r$ and all the nodes with marker 2 get the $\Delta$-learning tags as follows. The nodes with marker 2 are neighbors $w_1, w_2, \dots, w_{\lfloor\log \Delta\rfloor+1}$ of the node $r$. For each $i$, $1\le i\le \lfloor\log \Delta\rfloor+1$, node $w_i$ is assigned the tag $(B(i),b_i)$, where $B(i)$ is the binary representation of the integer $i$ and $b_i$ is the $i$-th bit of the binary representation of $\Delta$. The node $r$ gets the tag $(B,0)$, where $B$ is the binary representation of the integer $\lfloor\log \Delta\rfloor +1$. All other nodes of $G$ get the $\Delta$-learning tag $(0,0)$.
\item The collision tags will be used to create collisions. For each $l$, $0 \le l \le h-1$, each node in $V(l+1)$ gets the collision tag as follows. Let $US(l)=\{v_1,v_2,\cdots ,v_k\}$. For $1 \leq i \leq k$ and $1\le j\le \lfloor \log |N'(v_i,US(l))|\rfloor+1$, the node $u_j^i \in V(l+1)$ gets the collision tag $(id(u_j^i), b_m)$, where $m$ is the position of the integer $id(u_j^i)$ in the set $ID(v_i)$ in increasing order, and $b_m$ is the $m$-th bit of the binary representation of $|N'(v_i,US(l))|$. All other nodes $v \in V(l+1)$ get the collision tag $(0,0)$.
\item The weight-transmission tags will be used by nodes to transmit their weight to a unique node in the previous level. For each $l$, $0 \le l \le h-1$, each node in $V(l+1)$ gets the transmission tag as follows. Let $US(l)=\{v_1,v_2,\cdots, v_k\}$. For $1\le i\le k$, let $Q_i(x)=\{u \in N'(v_i,US(l))| W(u)=x\}$.  Choose any subset
$\{w_1,w_2,\dots, w_{ \lfloor \log |Q_i(x)|\rfloor+1}\}$ of $Q_i(x)$.  For $1 \le i \le  \lfloor \log |Q_i(x)|\rfloor+1$, the node $w_i$ gets the weight-transmission tag $(B(i),b_i)$, where $B(i)$ is the binary representation of the integer $i$, and $b_i$ is the $i$-th bit of the binary representation of $|Q_i(x)|$. All other nodes $v \in V(l+1)$ get the weight-transmission tag $(0,0)$.
\end{enumerate}
This completes the description of the labeling scheme $\Lambda$.

\subsubsection{Description of Algorithm {\tt Size Discovery}}

Algorithm {\tt Size Discovery} using the scheme $\Lambda$
consists of three procedures, namely Procedure {\tt Parameter Learning}, Procedure {\tt Size Learning},  and Procedure {\tt Final}.
The high-level idea and the detailed descriptions of each of these procedures are given below.

Procedure {\tt Parameter Learning}. The aim of this procedure is for every node in $G$ to learn three integers: $\Delta$, the number of the level to which the node belongs,
and  $h$. The procedure consists of two stages. In the first stage, that starts in round 1, every node with $M(2)=1$ and $M(0) =0$ (i.e., a neighbor of $r$ with marker 2) transmits its $\Delta$-learning tag in round $i$, if the id in the first component of this tag is $i$. The node with $M(0)=1$, i.e., the node  $r$, collects all the tags until it received a message from a node which has the same id as the id of $r$ in the $\Delta$-learning tag. After receiving this message, the node $r$ has learned all pairs $(B(1),b_1)$, ..., $(B(m), b_m)$, where $m$ is the id of $r$ and $B(i)$ is the binary representation of the integer $i$, corresponding to the $\Delta$-learning tag
 at the respective nodes. Then node $r$ computes the string $s=(b_1b_2 \dots b_m)$. This is the binary representation of $\Delta$.

 In the second stage, after learning $\Delta$, the node $r$ initiates the subroutine {\em Wave}$(\Delta)$. Every node other than $r$ waits until it detects two consecutive non-silent rounds. This indicates the end of the wave at this node and happens $2m+2$ rounds after the wave has been  started by the nodes of the previous level. The node computes $s$, learns $\Delta$, computes $m=\lfloor\log \Delta\rfloor+1$, and sets its level number as $j$, if  the end of the wave at this node occurred in round $m+j(2m+2)$.

When the unique node with  $M(1)=1$ learns its level number (which is $h$), it transmits the value of $h$ in the next round.
After receiving the first message containing an integer, a node with $M(3)=1$ sets $h$ to this integer and retransmits it. When the node with $M(0)=1$, i.e., the node $r$, gets the first message after round $m$ that contains an integer, it learns $h$ and initiates $Wave(h)$. The stage and the entire procedure end in round $t_1=m+h(2m+2) +h + h(2( \lfloor \log h \rfloor +1)+2)$. Note that after learning $h$, every node can compute $t_1$ and thus knows when Procedure
{\tt Parameter Learning} ends.

Procedure {\tt Size Learning}. This is the crucial procedure of the algorithm.  Its aim is to learn the size of the graph by the node $r$, i.e., to learn its weight $W(r)$. This procedure consists of $h$ phases.
In the $i$-th phase,  where $1\le i\le h$,  the participating nodes are from level $h-i+1$ and from level $h-i$. We will show by induction on $i$ that at the end of  the $i$-th phase,
all nodes of level $h-i$ correctly compute their weights. Thus at the end of the $h$-th phase, the node $r$ will learn its weight, i.e.,  the size of the network.
The high-level idea of the $i$-th phase is the following. In order to learn its weight, a node $v$ in $US(h-i)$ must learn the weights of all nodes $u$ in $N'(v,US(h-i))$ and subsequently add all these weights. Weight-transmission tags are used to achieve this.
The difficulty consists in preventing other neighbors in level $h-i$ of such nodes $u$ from adding these weights when computing their own weight, as this would result in multiple accounting (see Fig. \ref{fig:fig1}). This is done using collision tags to create collisions in other such nodes, so that nodes $z$ in $US(h-i)$ can identify neighbors in level $h-i+1$ outside of $N'(z,US(h-i))$ and ignore their weights.
A node transmits its weight-transmission tag in a round which is an increasing function of its weight. Since the nodes in $US(h-i)$ do not have any knowledge about their degree, they must learn the maximum possible weight of a node
in level $h-i+1$, to determine how long they must wait before receiving the last message from such a node.

We now give a detailed description of the $i$-th phase. At the beginning of the first phase, all nodes in level $h$ set their weight to 1.
The $i$-th phase starts in round $t_2(i)+1$, where $t_2(1)=t_1$, and ends in round $t_2(i+1)$. We will show that $t_2(i+1)$ will be known by every node of the graph by the end of the $i$-th phase, i.e., by the round $t_2(i+1)$.

In round $t_2(i)+1$ (which starts the $i$-th phase) the unique node $u'$ of level $h-i+1$ with $M(6)=1$ (which is a node of this level with maximum weight), initiates $Wave(W(u'))$. Every node in $G$ learns the value $x_i$ which is the maximal weight of a node in level $h-i+1$, by round $t'_2(i)=t_2(i)+2h (2(\lfloor\log x_i\rfloor +1)+2)$.
Since every node knows $h$ and $t_2(i)$, and it learns $x_i$ during the wave subroutine, it can compute the value $t_2'(i)$ by which $Wave(W(u'))$ is finished.
After learning this integer, every node in level $h-i+1$ and every node in level $h-i$  maintains a variable $status$ which can be either $complete$ or $incomplete$. (The variable $status$ is proper to a particular phase. In what follows we consider $status$ for phase $i$.) Initially the status of every node in level $h-i$  with $M(4)=1$ (the nodes in $US(h-i)$) and of every node in level $h-i+1$ is $incomplete$. The initial status of the nodes in level $h-i$ with $M(4)=0$ is $complete$.

At any time, only incomplete nodes will participate in this phase. The nodes with $M(4)=0$, i.e., the nodes outside $US(h-i)$, set their weights to 1 and never participate in this phase.

After learning its weight, a node $v$ in $US(h-i)$ gets status $complete$ and transmits a stop message in a special round. All the nodes in $N'(v,US(h-i))$ learn this stop message either by receiving it or by detecting a collision in this special round, and become $complete$. Thus the nodes in $N'(v,US(h-i))$ never transmit in subsequent rounds, and this prevents multiple accounting of the weights.

Let $z$ be a node in level $h-i+1$ with status $incomplete$.

If the id in the collision tag of $z$ is a positive integer $e$, then $z$ performs the following steps.

\begin{itemize}
\item The node $z$ transmits its collision tag for the first time in the $i$-th phase in round $t'_2(i)+e$.
After that, the node $z$ transmits its collision tag in every round $t'_2(i)+e+j\tau_i$, where  $\tau_i=\lfloor\log \Delta \rfloor+1+x_i(\lfloor\log \Delta\rfloor+1)+1$,
and $j\geq 1$, until it gets a stop message or detects a collision in round $t'_2(i)+j'\tau_i$, for some integer $j' \ge 1$.
In the latter case, node $z$ updates its status to $complete$.
\end{itemize}

If the id in the weight-transmission tag of $z$ is a positive integer $e'$, then $z$ performs the following steps.

\begin{itemize}
\item The node $z$ transmits the pair $(t,W(z))$, where $t$ is its weight-transmission tag and $W(z)$ is its weight, for the first time in the $i$-th phase in round $t'_2(i)+(W(z)-1)(\lfloor\log \Delta \rfloor+1)+e'$. After that the node $z$ transmits  $(t,W(z))$ in every round  $t'_2(i)+(W(z)-1)(\lfloor\log \Delta \rfloor+1)+e'+j\tau_i$, where  $\tau_i=\lfloor\log \Delta \rfloor+1+x_i(\lfloor\log \Delta\rfloor+1)+1$,
and $j\geq 1$, until it gets a stop message or detects a collision in the round  $t'_2(i)+j'\tau_i$ for some integer $j' \ge 1$. In the latter case, it updates its status to $complete$.
\end{itemize}

Let $z'$ be a node with $M(4)=1$, i.e., a node in $US(h-i)$. The node $z'$ (with status {\em incomplete}) performs the following steps.

\begin{itemize}
\item If $z'$ does not detect any collision in the time interval $[t'_2(i)+(j-1)\tau_i, t'_2(i)+(j-1)\tau_i+ \lfloor \log \Delta \rfloor+1]$, for some integer $j \ge 1$, then the node changes its status to $complete$. In this interval, the node $z'$ received the collision tags from the nodes in $N'(z',US(h-i))$.
    Suppose that the node $z'$ learns the pairs $(B(g_1),b_1)$, $(B(g_2),b_2)$, $\cdots$ , $(B(g_k), b_k)$, where $B(g_1), B(g_2), \cdots, B(g_k)$ are the binary representations of the integers $g_1, g_2, \cdots,g_k$, respectively, in the increasing order, corresponding to the collision tags of the respective nodes.
    The node $z'$ computes $s'=(b_1b_2\cdots b_k)$. Let $d$ be the integer whose binary representation is $s'$.
    The integer $d$ is the size of $N'(z',US(h-i))$.
    Then $z'$ waits until round $t_2'(i)+j\tau_i$.
    By this time, all nodes in level $h-i+1$ that transmitted according to their collision tags and weight-transmission tags, have already completed all these transmissions.
    If $z'$ detects any collision in the time interval $[t'_2(i)+(j-1)\tau_i+ \lfloor \log \Delta \rfloor+2,t_2'(i)+j\tau_i-1]$,
    it changes its status back to $incomplete$. Otherwise, for $1 \le f \le x_i$, let $(B(1),b_1), (B(2),b_2), \cdots, (B(g(f)), b_{g(f)})$ be the weight-transmission tags that the node $z'$ received from a node with weight $f$, where $B(a)$ is the binary representation of the integer $a$. Let $s_f=(b_1b_2\cdots b_{i(f)})$ and let $d_f$ be the integer whose binary representation is $s_f$. The integer $d_f$ is the total number of nodes of weight $f$ in $N'(z',US(h-i))$.

   The node $z'$ computes the value $\sum_f d_f$. If the node $z'$ had received any message from a node which is not in $N'(z',US(h-i))$, then the sum
    $\sum_f d_f$ cannot be equal to the integer $d$, and hence the node learns that there is a danger of multiple accounting of weights. In that case,
    the node changes its status back to $incomplete$.

    Otherwise, if $\sum_f d_f=d$, node $z'$ assigns $W(z')=1+\sum_f (fd_f)$. After computing $W(z')$, the node $z'$ transmits a stop message in round $t_2'(i)+j\tau_i$. If $z'$ is the node with $M(5)=1$ (i.e., the last node of $US(h-i)$), then after sending the stop message, it initiates $Wave(T)$, where $T$ is the current round number.  After learning $T$ from $Wave(T)$, every
    node in $G$ computes $t_2(i+1)=T+2h(2(\lfloor \log T\rfloor +1)+2)$. This is the round by which $Wave(T)$ is finished.
    In this round, the $i$-th phase of the procedure is finished as well.
\end{itemize}
At the end of the $h$-th phase, the node $r$ learns its weight, sets $n=W(r)$ and the procedure ends.

Procedure {\tt Final}: After computing $n$, the node $r$ initiates $Wave(n)$. Every node in $G$ computes the value of $n$, and outputs it. The procedure ends after all nodes output $n$.

Now our algorithm can be succinctly formulated as follows:

\begin{algorithm}
\caption{\textsc{{\tt Size Discovery}}}
\begin{algorithmic}[1]
\STATE {\tt Parameter Learning}
\STATE {\tt Size Learning}
\STATE {\tt Final}
\end{algorithmic}\label{alg:main}
\end{algorithm}

\subsection{Correctness and analysis}

The proof of the correctness of Algorithm {\tt Size Discovery} is split into two lemmas.
\begin{lemma}
Upon completion of the Procedure {\tt Parameter Learning}, every node in $G$ correctly computes $\Delta$, $h$, and its level number.
Moreover, every node computes the round number $t_1=m+h(2m+2) +h + h(2( \lfloor \log h \rfloor +1)+2)$ by which the procedure is over.
\end{lemma}
\begin{proof}
After round $m=\lfloor \log \Delta \rfloor +1 $, the node $r$ learns all pairs $(B(1),b_1)$, ..., $(B(m), b_m)$, where $B(i)$ is the binary representation of the integer $i$, corresponding to the $\Delta$-learning tags at the respective nodes with $M(2)=1$.
According to the assignment of the tags to the nodes, the binary string $s=(b_1b_2\dots b_m)$ is the binary representation of the integer $\Delta$. Therefore, the node $r$ correctly learns $\Delta$.

After learning $\Delta$, the node $r$ initiate $Wave(\Delta)$. For every level $i\ge 1$, the wave ends at the nodes in level $i$ in round $m+i(2m+2)$. The nodes learn the value of $\Delta$ from the wave and calculate their level number. The node in the $h$-th level for which $M(1)=1$, learns $h$ and transmits the value of $h$ along the path with nodes for which $M(3)=1$. The node $r$ learns $h$, and initiate $Wave(h)$. Every node computes $h$ from the wave.  Knowing $m$ and $h$ every node computes $t_1$.
\end{proof}

\begin{lemma}\label{main}
At the end of the $i$-th phase of the Procedure {\tt Size Learning}, every node in level $h-i$ correctly computes its weight.
\end{lemma}
\begin{proof}
We prove this lemma in two steps. First, we prove the following two claims, and then we prove the lemma by induction using these claims.

{\bf Claim 1:}  In the $i$-th phase of Procedure {\tt Size learning}, if the status of a node  $v_p\in US(h-i)$ is changed from $incomplete$ to $complete$ in the time interval $[t'_2(i)+(j-1)\tau_i,t'_2(i)+j\tau_i-1]$, for some integer $j\ge 1$, and remains $complete$ forever, then the node $v_p$ correctly computes $W(v)$ in round $t_2(i)+j\tau_i$, provided that all nodes of level $h-i+1$ know their weight at the beginning of the $i$-th phase.

In order to prove this claim, suppose that the status of $v_p\in US(h-i)$ is changed to $complete$ from $incomplete$ in the interval $[t'_2(i)+(j-1)\tau_i,t'_2(i)+j\tau_i-1]$, for some integer $j\ge 1$. Since the status of $v_p$ is $complete$ in round $t_2(i)+j\tau_i$, the node $v_p$ did not detect any collision in the above time interval.
Suppose that $v_p$ received messages only from nodes in $N'(v_p,US(h-i))$. In round $t'_2(i)+(j-1)\tau_i+\lfloor\log\Delta\rfloor+1$, the node computes the integer $d$ from the collision tags of the nodes from which it received messages in the time interval $[t'_2(i)+(j-1)\tau_i, t'_2(i)+(j-1)\tau_i+\lfloor\log\Delta\rfloor+1]$. According to the labeling scheme, the bits in the collision tags of the nodes in $N'(v_p,US(h-i))$ were assigned in such a way that the string $s'$ formed by these bits is the binary representation of the integer $|N'(v_p,US(h-i))|$. After that, the nodes in $N'(v_p,US(h-i))$ whose weight-transmission tag contains a positive integer as the id, transmit their tags to $v_p$ one by one. Let $X_f \subseteq N'(v_p,US(h-i))$ be the the set of nodes in $N'(v_p,US(h-i))$ with weight $f$, for $1\le f\le x_i$. The weight-transmission tags are given to $(\lfloor\log |X_f|\rfloor+1)$ nodes in $X_f$ in such a way that the binary string formed by the bits of the weight transmission tags of these nodes in the increasing order of their ids is the binary representation of the integer $|X_f|$. Hence, for $1\le f\le x_i$, $\sum_f |X_f|=d$, as the sum of the numbers of nodes in $N'(v_p,US(h-i))$ with different weights is equal to the total number of nodes in $N'(v_p,US(h-i))$. If the node $v_p$ received messages only from the nodes in $N'(v_p,US(h-i))$, it learns $\sum_f |X_f|=d$, and hence correctly computes $W(v_p)=1+\sum_f f|X_f|$.

Otherwise,  there exists a node $u \in N(v_p)\cap N'(v_q,US(h-i))$, for some node $v_q \in US(h-i)$ with $q<p$, such that the id in the weight-transmission tag of $u$ is non-zero. Then the integer $\sum_f |X_f|$ that $v_p$ computes  cannot be equal to the integer $d$, as explained above, and the node
$v_p$ changes its status back to $incomplete$. This is a contradiction. Therefore, the node $v_p$ correctly  computes its weight at the end of round $t'_2(i)+j\tau_i-1$,
which proves the claim.

\noindent
{\bf Claim 2:} Let $US(h-i)=\{v_1,v_2,\cdots,v_k\}$.  In the $i$-th phase of Procedure {\tt Size Learning}, each node $v_j$ changes its status from $incomplete$ to $complete$
during the time interval $[t'_2(i)+(q_j-1)\tau_i+1,t'_2(i)+q_j\tau_i]$, for some $q_j \leq j$, and remains $complete$ forever.

 We prove this claim by induction on $j$. As the base case, we prove that in the time interval $[t'_2(i)+1,t'_2(i)+\tau_i-1]$, the status of the node $v_1 \in US(h-i)$ is changed from $incomplete$ to $complete$. According to the labeling scheme and to the construction of the set $US(h-i)$, $(\lfloor \log |N'(v_1,US(h-i))|\rfloor+1)$ nodes from $N'(v_1,US(h-i))$ have distinct positive ids in their collision tags, and all other nodes from $N'(v_1,US(h-i))$ have the id 0. Hence, the node $v_1$ detects no collision in the time interval  $[t'_2(i)+1,t'_2(i)+(\lfloor\log \Delta\rfloor+1)]$, and it changes its status to $complete$.  In the next $x_i(\lfloor\log \Delta \rfloor+1)$ rounds, the nodes of level $h-i
 +1$, with positive ids in their weight-transmission tags, transmit. Since the ids in the weight-transmission tags of $(\lfloor \log |N'(v_1,US(h-i))|\rfloor+1)$ nodes are distinct positive integers, and $N(v_1)=N'(v_1,US(h-i))$, the node $v_1$ does not detect any collision. Also, since the node $v_1$ received messages only from nodes in $N'(v_1,US(h-i))$,
 therefore $\sum_f |X_f|=d$, for $1\le f\le x_i$, and hence $v_1$ remains $complete$ forever.

Suppose by induction that Claim 2 holds for nodes $v_1,\dots ,v_j$.  Let $y=\max\{q_1,q_2,\dots, q_j\}$. Consider the following two cases.

Case 1: There exists an integer $q_{j+1} \leq y$, such that the node $v_{j+1}$ changes its status from $incomplete$ to $complete$
during the time interval $[t'_2(i)+(q_{j+1}-1)\tau_i+1,t'_2(i)+q_{j+1}\tau_i]$, for some $q_{j+1} \leq y$, and remains $complete$ forever.

In this case the claim holds for $v_{j+1}$ because $q_{j+1} \leq y \leq j+1$.

Case 2: Case 1 does not hold.

Therefore, the status of $v_{j+1}$ is $incomplete$ in round  $t'_2(i)+y\tau_i$.
The status of all the nodes in $N'(v_{j+1},US(h-i))$ is $incomplete$ in this round as well, as they did not received any stop message from $v_{j+1}$ or detected any collision in round $t'_2(i)+y\tau_i$.

 The status of the nodes in $N(v_{j+1})\setminus N'(v_{j+1},US(h-i))$ is $complete$, as $N(v_{j+1})\setminus N'(v_{j+1},US(h-i)) \subseteq \cup_{i=1}^{j} N(v_i)$ and the nodes $v_1,v_2,\cdots,v_j$ are $complete$. Consider the time interval $[t'_2(i)+y\tau_i+1,t'_2(i)+(y+1)\tau_i-1]$.
In this time interval, the node $v_{j+1}$ receives messages only from the nodes in $N'(v_{j+1},US(h-i))$. Since the positive ids in the collision tags and the positive ids in the weight-transmission tags are unique for the nodes in $N'(v_{j+1},US(h-i))$, the node $v_{j+1}$ does not detect any collision in the interval
$[t'_2(i)+y\tau_i+1,t'_2(i)+(y+1)\tau_i-1]$. Also, since the node $v_{j+1}$ received messages only from nodes in $N'(v_{j+1},US(h-i))$,
 therefore $\sum_{f=1}^{x_i} d_f=d$, and hence $v_{j+1}$ remains $complete$ forever. Since $y \le j$, we have $y+1\le j+1$.
Therefore, the proof of the claim follows by induction.

%

Now we prove the lemma by induction on the phase number. According to the definition of the weight of a node, all the nodes in level $h$ have weight 1. Therefore, by Claim 2, at the end of round $t_2(1)+j\tau_1$, the node $v_j$ in $US(h-1)$ becomes $complete$, and hence by Claim 1, it correctly computes its weight, since all the nodes in level $h$ already know their weight which is 1. This implies that all the nodes in level $h-1$ correctly compute their weights at the end of phase 1. Suppose that for $i\ge 1$, all the nodes in level $h-i$ correctly compute their weights at the end of phase $i$. Then by Claim 2, all the nodes in $US(h-i-1)$ become $complete$ in the ($i+1$)-th phase, and hence by Claim 1 they correctly compute their weights in this phase. Therefore, the lemma follows by induction.
\end{proof}

Applying Lemma \ref{main} for $i=h$, we get the following corollary.
\begin{corollary}\label{cor}
Upon completion of Procedure {\tt Size Learning}, the node $r$ correctly computes the size of the graph.
\end{corollary}

Now we are ready to formulate our main positive result.

\begin{theorem}
The length of the labeling scheme used by Algorithm {\tt Size Discovery} on a graph of maximum degree $\Delta$ is $O(\log\log \Delta)$.
Upon completion of this algorithm, all nodes correctly output the size of the graph.
\end{theorem}
\begin{proof}
According to the labeling scheme $\Lambda$, the label of every node has two parts. The first part is a vector $M$ of constant length and each term of $M$ is one bit. The second part is a vector $L$ containing three tags, each of which is of length $O(\log \log \Delta)$. Therefore, the length of the labeling scheme $\Lambda$ is $O(\log\log \Delta)$.

By Corollary \ref{cor}, node $r$ correctly computes the size $n$ of the network upon completion of Procedure {\tt Size Learning}. In Procedure {\tt Final}, node $r$ initiates
$Wave(n)$, and hence every node correctly computes $n$ upon completion of the algorithm.
\end{proof}

\subsection{The lower bound}
%
%
%
%
%
%

In this section, we show that the length of the labeling scheme used by Algorithm {\tt Size Discovery} is optimal, up to multiplicative constants. We prove the matching lower bound by showing that for some class of graphs of maximum degree $\Delta$ (indeed of trees), any size discovery algorithm must use a labeling scheme of length at least $\Omega(\log \log\Delta)$ on some graph of this class.

Let $S$ be a star with the central node $r$ of degree $\Delta$. Denote one of the leaves of $S$ by $a$.
For $\lfloor \frac{\Delta}{2}\rfloor\le i \le \Delta-1 $, we construct a tree $T_i$ by attaching $i$ leaves to $a$. The maximum degree of each tree $T_i$ is $\Delta$. Let $\cT$ be the set of trees $T_i$, for  $\lfloor \frac{\Delta}{2}\rfloor\le i \le \Delta-1 $, cf. Fig. \ref{fig:fig2}. Hence the size of $\cT$ is at least $\frac{\Delta}{2}$.

\begin{figure}[h]
\centering
\includegraphics[width=0.5\textwidth]{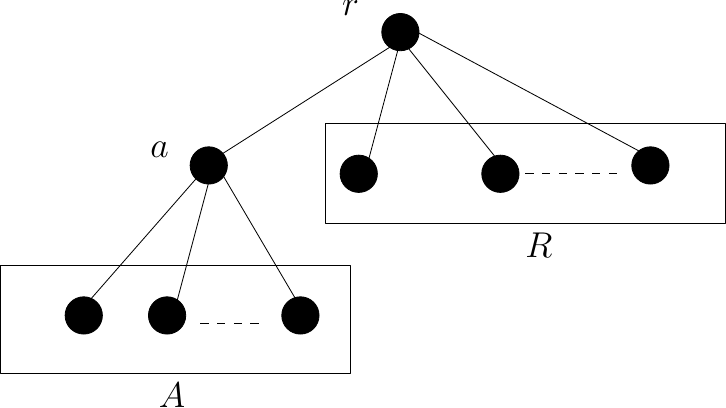}
\caption{Example of a tree in $\cal T$ }
\label{fig:fig2}
\end{figure}

The class $\cT$ of trees was used in \cite{GP} to prove an analogous lower bound for the problem of topology recognition (which, for the class $\cT$, is equivalent to size discovery). However, it should be stressed that the proof of the lower bound in our present scenario is much more involved because we work under the more powerful model assuming the capability of collision detection, while \cite{GP} assumed no collision detection. The negative result under our more powerful model is more difficult to obtain because of potential possibility of acquiring information by nodes from hearing collisions.
More precisely, our negative argument is based on the fact that in a deterministic algorithm nodes with the same history (see the formal definition below) must behave identically. In the model with collision detection, histories are more complicated because they are composed not only of messages heard by nodes in previous rounds but also of collisions heard by them.

Let $R$ be the set of leaves attached to $r$ and let $A$ be the set of leaves attached to $a$. For a tree $T\in \cT$, consider a labeling scheme $L(T)$ of length $\beta$, and let $\cA$ be an algorithm that finds the size of every tree $T \in \cT$, using $L(T)$. Let  $L(T)$  assign the label $l(v)$ to each node $v$ in $T$.

Let $T \in \cT$ be any tree. We define the notion of history (a similar notion was defined in \cite{Pe} for anonymous radio networks without collision detection) for each node $v$ in $T$ in round $t$. The history of a node in time $t$ is denoted by $H(v,t,L,\cA)$. This is the information that node $v$ acquires by round $t$, using the algorithm $\cA$. The action of a node  $v$ in round $t+1$ is a function of the history $H(v,t,L,\cA)$, hence for every round $t$, if two nodes have the same history in round $t$, then they behave identically in round $t+1$.
As in \cite{Pe}, we assume without loss of generality, that whenever a node transmits a message in round $t+1$, it sends its entire history in round $t$.
We define the history by induction on the round number as follows. $H(v,0,L,\cA)=l(v)$, for each node $v$ in $T$. For $t\ge 0$, the history in time $t+1$ is defined as follows, using the histories of the nodes in $T$ in time $t$.

\begin{itemize}
\item If $v$ receives a message from a node $u$ in round $t+1$, i.e., $v$ is silent in this round, and $u$ is its only neighbor that transmits in this round, then $H(v,t+1,L,\cA)=[H(v,t,L,\cA),H(u,t,L,\cA)]$.
\item If $v$ detects a collision in round $t+1$, i.e., $v$ is silent in this round, and there are at least two neighbors of $v$ that transmit in this round, then  $H(v,t+1,L,\cA)=[H(v,t,L,\cA),*]$.
\item Otherwise, $H(v,t+1,L,\cA)=[H(v,t,L,\cA),\lambda]$.
\end{itemize}

Hence, histories are nested sequences of labels and of symbols $\lambda$, and $*$, where, intuitively, $\lambda$ stands for silence in a given round, and $*$ stands for a collision.

The following lemma shows that histories of nodes  in sets $A$ and $R$ are equal iff the labels of these nodes are the same.

\begin{lemma}\label{lem:lem1}
For any tree $T \in \cT$ consider a labeling scheme $L(T)$. Let $\cA$ be any algorithm that finds the size of every tree $T \in \cT$ using the scheme $L(T)$. Then for any $t\ge 0$, we have:

\begin{enumerate}
\item For $v_1,v_2 \in R$, $H(v_1,t,L,\cA)=H(v_2,t,L,\cA)$, if and only if $l(v_1)=l(v_2)$.
\item For $v_1,v_2 \in A$, $H(v_1,t,L,\cA)=H(v_2,t,L,\cA)$, if and only if $l(v_1)=l(v_2)$.
\end{enumerate}
\end{lemma}

\begin{proof}
We prove the first part of the lemma. The proof of the second part is similar.
By definition, for two nodes $v_1$ and $v_2$ with different labels, we have $H(v_1,t,L,\cA)\ne H(v_2,t,L,\cA)$ for all $t \ge 0$.

To prove the converse, we use induction on $t$. Let $v_1,v_2\in R$ such that $l(v_1)=l(v_2)$.
For $t=0$, $H(v_1,0,L,\cA)=l(v_1)=l(v_2)=H(v_2,0,L,\cA)$. Suppose that the statement is true for round $t$, i.e., $H(v_1,t,L,\cA)=H(v_2,t,L,\cA)$. Note that the history
of any node in $R$ in round $t+1$ does not depend on any action performed by the node $a$ or the nodes in $A$ in round $t+1$. Also, since the nodes $v_1$ and $v_2$ have the same histories in round $t$, they must behave identically in round $t+1$. Therefore, in round $t+1$, there can only be the following four cases.
\begin{enumerate}
\item[Case 1] The node $r$ transmits and the nodes $v_1,v_2$ do not transmit.

According to the definition of history, $H(v_1,t+1,L,\cA)=[H(v_1,t,L,\cA), H(r,t,L,\cA)]$ and $H(v_2,t+1,L,\cA)=[H(v_2,t,L,\cA), H(r,t,L,\cA)]$. This implies $H(v_1,t+1,L,\cA)=H(v_2,t+1,L,\cA)$.

 \item[Case 2] The nodes $v_1$ and $v_2$ transmit and the node $r$ does not transmit.

 According to the definition of history,
     $H(v_1,t+1,L,\cA)=[H(v_1,t,L,\cA),\lambda]$ and $H(v_2,t+1,L,\cA)=[H(v_2,t,L,\cA),\lambda]$. This implies $H(v_1,t+1,L,\cA)=H(v_2,t+1,L,\cA)$.
 \item[Case 3] The nodes $v_1$ and $v_2$, and the node $r$ transmit.

 According to the definition of history,
     $H(v_1,t+1,L,\cA)=[H(v_1,t,L,\cA),\lambda]$ and $H(v_2,t+1,L,\cA)=[H(v_2,t,L,\cA),\lambda]$. This implies $H(v_1,t+1,L,\cA)=H(v_2,t+1,L,\cA)$.
\item[Case 4] Neither $v_1$ nor $v_2$ nor $r$ transmit. In this case, $H(v_1,t+1,L,\cA)=[H(v_1,t,L,\cA),\lambda]$ and $H(v_2,t+1,L,\cA)=[H(v_2,t,L,\cA),\lambda]$. This implies $H(v_1,t+1,L,\cA)=H(v_2,t+1,L,\cA)$.
\end{enumerate}

Hence the proof of the lemma follows by induction.
\end{proof}

With the length of the labeling scheme $\beta$, there can be at  most $z=2^{\beta+1}$ possible different labels of at most this length. Let $\cL=\{l_1,l_2,\cdots,l_z\}$ be the set of distinct labels of length at most $\beta$. We define the {\it pattern} of a tree $T$ with the labeling scheme $L(T)$ as the pair $(P(r),P(a))$, where $P(r)$ and $P(a)$ are defined as follows.

$ P(r)=(l(r),b_1,b_2, \cdots, b_z)$, where $b_i\in \{0,1,2\}$ and:\\
 $b_i=0$, if no node in $R$ has label $l_i$;\\
  $b_i=1$, if there is exactly one node in $R$ with label $l_i$;\\
  $b_i=2$, if there are more than one node in $R$ with label $l_i$.

$ P(a)=(l(a),b'_1,b'_2, \cdots, b'_z)$, where $b'_i\in \{0,1,2\}$ and:\\
 $b'_i=0$, if no node in $A$ has label $l_i$;\\
 $b'_i=1$, if there is exactly one node in $A$ with label $l_i$;\\
 $b'_i=2$, if there are more than one node in $A$ with label $l_i$.

 The following lemma states that histories of the node $r$ in trees from $\cT$ depend only on the pattern and not on the tree itself.

\begin{lemma}\label{lem:lem2}
Let $\cA$ be any algorithm that solves the size discovery problem for all trees $T\in \cT$ using the labeling scheme $L(T)$.
 If trees $T_1$ and $T_2$ have the same pattern, then for any $t \ge 0$, the node $r$ in $T_1$ and the node $r$ in $T_2$ have the same history in round $t$.
\end{lemma}
\begin{proof}
Let $T_1$ and $T_2$ be two trees with same pattern $(P(r),P(a))$. For $j=1,2$,
denote the node $r$ in $T_j$ by $r_j$, the node $a$ in $T_j$ by $a_j$, the set $R$ in $T_j$ by $R_j$, and the set $A$ in $T_j$ by $A_j$.
For any $t \ge 0$, we prove the following statements by simultaneous induction. (To prove the lemma, we need only the first of them).

\begin{enumerate}
\item $H(r_1,t,L,\cA)=H(r_2,t,L,\cA)$.
\item  $H(a_1,t,L,\cA)=H(a_2,t,L,\cA)$.
\item For a node $v_1$ in $R_1$ and a node $v_2$ in $R_2$ with same label, $H(v_1,t,L,\cA)=H(v_2,t,L,\cA)$.
\item  For a node $v_1$ in $A_1$ and a node $v_2$ in $A_2$ with the same label, $H(v_1,t,L,\cA)=H(v_2,t,L,\cA)$.
\end{enumerate}

Since the patterns of the two trees are the same, we have $l(r_1)=l(r_2)$, and $l(a_1)=l(a_2)$. Therefore, according to the definition of the history, the above statements are true for $t=0$.

Suppose that all the above statements are true for round $t$. Consider the execution of the algorithm in round $t+1$ as follows:

\noindent
{\bf Induction step for (1)}: The actions of  nodes in $A_1$ and $A_2$ in round $t+1$ do not affect the histories of the nodes $r_1$ and $r_2$ in this round.  Hence we have the following cases in round $t+1$.
\begin{enumerate}[(a)]
\item $r_1$ transmits, $a_1$ does not transmit, and no node in $R_1$ transmits.

According to the definition of history, $H(r_1,t+1,L,\cA)=[H(r_1,t,L,\cA),\lambda]$. Since $H(r_1,t,L,\cA)=H(r_2,t,L,\cA)$ and $r_1$ transmits in round $t+1$, then $r_2$ also transmits in round $t+1$. Similarly $a_2$ does not transmit in this round, since $a_1$ does not transmit. We prove that no node in $R_2$ transmits in round $t+1$. Suppose otherwise. Let $v_2$ be a node in $R_2$ that transmits in round $t+1$,
 and suppose that the label of $v_2$ is $l_i$. Therefore, in $P(r)$, either $b_i=1$, or $b_i=2$.

 If $b_i=1$, then $v_2$ is the unique node with label $l_i$ in $R_2$. Since the patterns of the two trees are the same, therefore, there exists a unique node $v_1$ in $R_1$ with label $l_i$. Since $H(v_1,t,L,\cA)=H(v_2,t,L,\cA)$, by the induction hypothesis for (3), therefore $v_1$ must transmit in round $t+1$, which is a contradiction with the fact that no node in $R_1$ transmits.
 A similar statement holds for $b_i=2$.
    Hence no node in $R_2$ transmits in round $t+1$. Therefore, $H(r_2,t+1,L,\cA)= [H(r_2,t,L,\cA),\lambda]$. This implies that $H(r_1,t+1,L,\cA)=H(r_2,t+1,L,\cA)$.

\item $r_1$ transmits, $a_1$ transmits, no node in $R_1$ transmits.

Since $H(r_1,t,L,\cA)=H(r_2,t,L,\cA)$ and $H(a_1,t,L,\cA)=H(a_2,t,L,\cA)$, and $r_1$ and $a_1$ transmit in round $t+1$, therefore, $r_2$ and $a_2$ also transmit in round $t+1$. Hence $H(r_1,t+1,L,\cA)=[H(r_1,t,L,\cA),\lambda]$ and $H(r_2,t+1,L,\cA)=[H(r_2,t,L,\cA),\lambda]$. This implies that $H(r_1,t+1,L,\cA)= H(r_2,t+1,L,\cA)$.

\item $r_1$ transmits, $a_1$ does not transmit, some nodes in $R_1$ transmit.

Similarly as in (b), we have $H(r_1,t+1,L,\cA)=[H(r_1,t,L,\cA),\lambda]$ and $H(r_2,t+1,L,\cA)=[H(r_2,t,L,\cA),\lambda]$. This implies that $H(r_1,t+1,L,\cA)= H(r_2,t+1,L,\cA)$.

\item $r_1$ transmits, $a_1$ transmits, some nodes in $R_1$ transmit.

Similarly as in (b), we have $H(r_1,t+1,L,\cA)=[H(r_1,t,L,\cA),\lambda]$ and $H(r_2,t+1,L,\cA)=[H(r_2,t,L,\cA),\lambda]$. This implies that $H(r_1,t+1,L,\cA)= H(r_2,t+1,L,\cA)$.

\item $r_1$  does not transmit, $a_1$ does not transmit, no node in $R_1$ transmits.

Since $a_1$ does not transmit in round $t+1$, therefore $a_2$ does not transmit in round $t+1$. Also, as explained in (a), no node in $R_2$  transmits in round $t+1$. Therefore, $H(r_1,t+1,L,\cA)= [H(r_1,t,L,\cA),\lambda]$ and $H(r_2,t+1,L,\cA)= [H(r_2,t,L,\cA),\lambda]$. This implies that $H(r_1,t+1,L,\cA)= H(r_2,t+1,L,\cA)$.

\item $r_1$  does not transmit, $a_1$ transmits, no node in $R_1$ transmits.

Since $a_1$ transmits in round $t+1$, therefore $a_2$ transmits in round $t+1$. Also, as explained in (a), no node in $R_2$ transmits in round $t+1$. Therefore, according to the definition of history, $H(r_1,t+1,L,\cA)=[H(r_1,t,L,\cA),[H(a_1,t,L,\cA)]$, and $H(r_2,t+1,L,\cA)=[H(r_2,t,L,\cA),[H(a_2,t,L,\cA)]$. Since   $H(r_1,t,L,\cA)=H(r_2,t,L,\cA)$ and $H(a_1,t,L,\cA)=H(a_2,t,L,\cA)$, therefore, $H(r_1,t+1,L,\cA)=H(r_2,t+1,L,\cA)$.

\item $r_1$  does not transmit, $a_1$ does not transmit, some nodes in $R_1$ transmit.

Let $v_1$ be a node in $R_1$ with label $l_i$, such that $v_1$ transmits in round $t+1$. Then by Lemma \ref{lem:lem1}, all the nodes in $R_1$ with label $l_i$ must transmit in round $t+1$. Suppose that the nodes with labels  $l_{i_1},l_{i_2}, \cdots,l_{i_k}$ transmit in round $t+1$.

If each of the integers $b_{i_1},b_{i_2}, \cdots,b_{i_k}$ is 0, then no node in $R_1$ transmits  which contradicts the assumption of case (g).

    If at least two of the integers $b_{i_1},b_{i_2}, \cdots,b_{i_k}$ are 1, or at least one of them is 2, then there exist at least two nodes in $R_1$ and at least two nodes in $R_2$ that transmit in round $t+1$. Hence, a collision is heard at the node $r_1$ and a collision is heard at the node $r_2$. Therefore,  $H(r_1,t+1,L,\cA)=[H(r_1,t,L,\cA),*]$ and $H(r_2,t+1,L,\cA)=[H(r_2,t,L,\cA),*]$. This implies that $H(r_1,t+1,L,\cA)= H(r_2,t+1,L,\cA)$.

    Otherwise, exactly one of the integers $b_{i_1},b_{i_2}, \cdots,b_{i_k}$ is 1, and all others are 0. W.l.o.g.
    let $b_{i_1}$ be the unique  integer 1. Then there is exactly one node $v_1$ with label $l_{i_1}$ in $R_1$ and there is exactly one node $v_2$ with label $l_{i_1}$ in $R_2$ which transmit in round $t+1$. Therefore,
    $H(r_1,t+1,L,\cA)=[H(r_1,t,L,\cA),H(v_1,t,L,\cA)]$ and $H(r_2,t+1,L,\cA)=[H(r_2,t,L,\cA),H(v_2,t,L,\cA)]$. Since,
     $H(v_1,t,L,\cA)=H(v_2,t,L,\cA)$ and $H(r_1,t,L,\cA)=H(r_2,t,L,\cA)$,  therefore, $H(r_1,t+1,L,\cA)= H(r_2,t+1,L,\cA)$.

\item $r_1$  does not transmit, $a_1$ transmits, some nodes  in $R_1$ transmit.

Since $a_1$ transmits in round $t+1$, therefore, $a_2$ transmits in round $t+1$. Also, since some node in $R_1$ transmits in round $t+1$, therefore some node in $R_2$ transmits in round $t+1$, as explained in (a). Therefore, a collision is heard at $r_1$, and a collision is heard at $r_2$. Hence, $H(r_1,t+1,L,\cA)=[H(r_1,t,L,\cA),*]$ and $H(r_2,t+1,L,\cA)=[H(r_2,t,L,\cA),*]$. This implies that $H(r_1,t+1,L,\cA)= H(r_2,t+1,L,\cA)$.

\end{enumerate}

\noindent
{\bf Induction step for (2)}: This is similar to the induction step for (1).

\noindent
{\bf Induction step for (3)}: For $j=1,2$, the histories of the nodes in $R_j$ in round $t+1$ do not depend on the action of the  node $a_j$ and the actions of the nodes in $A_j$, in this round. Hence we have the following cases in round $t+1$.
\begin{enumerate}[(i)]
\item The node $r_1$ transmits and no node in $R_1$ transmits.

This implies that the node $r_2$ transmits and no node in $R_2$ transmits. Therefore, $H(v_1,t+1,L,\cA)=[H(v_1,t,L,\cA),H(r_1,t,L,\cA)]$ and $H(v_2,t+1,L,\cA)=[H(v_2,t,L,\cA),H(r_2,t,L,\cA)]$. Since,  $H(r_1,t,L,\cA)=H(r_2,t,L,\cA)$ and $H(v_1,t,L,\cA)=H(v_2,t,L,\cA)$, therefore, $H(v_1,t+1,L,\cA)=H(v_2,t+1,L,\cA)$.

\item The node $r_1$ transmits and some nodes in $R_1$ transmit.

This implies that the node $r_2$ transmits and some nodes in $R_2$ transmit.
There are two cases.
      If $v_1$ transmits, then $v_2$ also transmits in round $t+1$. Hence  $H(v_1,t+1,L,\cA)=[H(v_1,t,L,\cA),\lambda]$ and $H(v_2,t+1,L,\cA)=[H(v_2,t,L,\cA),\lambda]$ and hence $H(v_1,t+1,L,\cA)=H(v_2,t+1,L,\cA)$.
      If $v_1$ does not transmit, then $v_2$ does not transmit either, in round $t+1$. Hence  $H(v_1,t+1,L,\cA)=[H(v_1,t,L,\cA),H(r_1,t,L,\cA)]$ and
      $H(v_2,t+1,L,\cA)=[H(v_2,t,L,\cA),H(r_2,t,L,\cA)]$ and hence $H(v_1,t+1,L,\cA)=H(v_2,t+1,L,\cA)$.

\item The node $r_1$ does not transmit and some nodes in $R_1$ transmit.

Since $r_1$ does not transmit therefore $r_2$ does not transmit. According to the definition of history, $H(v_1,t+1,L,\cA)=[H(v_1,t,L,\cA),\lambda]$ and $H(v_2,t+1,L,\cA)=[H(v_2,t,L,\cA),\lambda]$, and hence $H(v_1,t+1,L,\cA)=H(v_2,t+1,L,\cA)$.

\item The node $r_1$ does not transmit and no node in $R_1$ transmits.

In this case,  $H(v_1,t+1,L,\cA)=[H(v_1,t,L,\cA),\lambda]$ and $H(v_2,t+1,L,\cA)=[H(v_2,t,L,\cA),\lambda]$ and hence $H(v_1,t+1,L,\cA)=H(v_2,t+1,L,\cA)$.

\end{enumerate}

\noindent
{\bf Induction step for (4)}: This is similar to the induction step for (3).

Therefore, the lemma follows by induction.
\end{proof}

\begin{corollary}
Let $\cH_t$ be the set of all possible histories of the node $r$ in all trees in $\cT$, in round $t$, and let $\cP$ be the set of all possible patterns of trees in $\cT$. Then $|\cH_t| \le |\cP|$.
\label{cor:cor1}\end{corollary}

The following theorem gives the lower bound $\Omega(\log \log \Delta)$ on the length of a labeling scheme for size discovery, that matches the length of the labeling scheme used by
Algorithm {\tt Size Discovery}.

\begin{theorem}
For any tree $T \in \cT$ consider a labeling scheme L($T$). Let $\cA$ be any algorithm that finds the size of $T$, for every tree $T \in \cT$, using the scheme $L(T)$. Then there exists a tree $T' \in \cT$, for which the length of the scheme $L(T')$ is $\Omega(\log \log \Delta)$.
\end{theorem}
\begin{proof}
It is enough to prove the theorem for sufficiently large $\Delta$.
We prove the theorem by contradiction. Suppose that there exists an algorithm $\cA$ that solves the size discovery problem in the class $\cT$, in time $t$, with labels of length at most $\frac{1}{4}\log \log \Delta$. There are at most $z=2{(\log \Delta)}^{\frac{1}{4}}$ possible different labels of at most this length.
There are at most $z^2 3^{2z}$ different possible patterns for these $z$ labels. Therefore, by Corollary \ref{cor:cor1}, the total number of histories of the node $r$, in time $t$, over the entire class $\cT$,  is at most $z^2 3^{2z} < (2{(\log \Delta)}^{\frac{1}{4}})^2 3^{4{(\log \Delta)}^{\frac{1}{4}}}<\frac{\Delta}{2} \le |\cT|$, for sufficiently large~$\Delta$.

Therefore, by the pigeonhole principle, there exist two trees $T'$, $T''$ in $\cT$ such that the history of $r$ in $T'$ in time $t$ is the same as the history of $r$ in $T''$ in round $t$. This implies that the node $r$ in $T'$ and the node $r$ in $T''$ must behave identically in every round until round $t$, hence they must  output the same size. This contradicts the fact that the trees $T'$ and $T''$ have different sizes.
This completes the proof.
\end{proof}

\section{Finding the diameter of a network}

By contrast to the task of finding the size of a network, it turns out that finding the diameter of a network can be done with a much shorter labeling scheme: in fact, we will show that a scheme of constant length is sufficient. In our solution the labeling scheme has length 2.

Consider any graph $G$ and let $u$ and $v$ be two nodes at distance equal to the diameter $D$ of the graph. The labeling scheme is as follows: node $u$ has label (00), node $v$ has label (11), all other  nodes at distance $D$ from $u$ have label (10), and all nodes at positive distances  $<D$ from node $u$ have label (01). Algorithm {\tt Diameter Discovery} consists of two procedures: Procedure {\tt Find Diameter} initiated by the node with label (00), and $Wave(D)$ initiated by the node with label (11). Upon completion of {\tt Find Diameter}, the node with label (11) learns the correct value of $D$ which it subsequently broadcasts to all other nodes using $Wave(D)$.

We now describe Procedure {\tt Find Diameter}. It is a modification of the subroutine $Wave$ described in Section 2. Its aim is for all nodes at distance $i$ from the initiating node $u$ with label (00),
to learn $i$. The set of these nodes is called level $i$. As soon as the node $v$ with label (11) learns its distance from $u$, it knows that this distance is $D$ and then it initiates $Wave(D)$ to spread this knowledge to all other nodes.

 Procedure {\tt Find Diameter} works in phases $i=0,1,\dots, D-1$. At the beginning of phase $i$, all nodes at levels at most $i$ know their level number and all nodes at levels larger than $i$ do not know it.
Phase $i$ of Procedure {\tt Find Diameter} is identical to any phase of $Wave(i+1)$ with nodes at level $i$ playing the role of blue nodes, nodes at levels smaller than $i$ playing the role of white nodes, and nodes at levels larger than $i$ playing the role of red nodes, with the modification that nodes at level $i$ initiate it only
 if their label starts with bit 0. (Hence nodes with labels (11) and (10) which are exactly nodes at distance  $D$ from $u$ do not initiate phase $D$, which would be useless). Upon completion of phase $i$, all nodes at level $i+1$ learn the value of $i+1$.

 Hence at the end of phase $D-1$ all nodes know their level number. Node with label (11) that knows its level number $D$, knows that this is the diameter of the graph. It initiates $Wave(D)$. At this point, all nodes have finished all transmissions of Procedure
 {\tt Find Diameter} and hence there are no interferences with transmissions of $Wave(D)$, which starts with node $v$ colored blue and all other nodes colored red.

 Procedure {\tt Spread}  that follows Procedure {\tt Find Diameter} can be described as follows. After computing $D$, node $v$ with label (11) initiates $Wave(D)$. Every node in $G$ computes the value of $D$, and outputs it. The procedure ends after all nodes output $D$. Now our algorithm can be succinctly formulated as follows:

 \begin{algorithm}
\caption{\textsc{{\tt Diameter Discovery}}}
\begin{algorithmic}[1]
\STATE {\tt Find Diameter}
\STATE {\tt Spread Diameter}
\end{algorithmic}\label{alg:diam}
\end{algorithm}

The above explanations prove the following proposition.

\begin{proposition}
The length of the labeling scheme used by Algorithm {\tt Diameter Discovery} on any graph is 2.
Upon completion of this algorithm, all nodes correctly output the diameter of the graph.
\end{proposition}

%

 \section{Conclusion}

  We established the minimum length $\Theta(\log\log \Delta)$ of a labeling scheme permitting to find the size of arbitrary radio networks of maximum degree $\Delta$, with collision detection, and we designed a size discovery algorithm using a labeling scheme of this length. For the task of diameter discovery, we showed an algorithm using a labeling scheme of constant length.
  Our algorithms heavily use the collision detection capability, hence the first open question is whether our results hold in radio networks without collision detection. Secondly, in this paper we were concerned only with the feasibility of the size and diameter discovery tasks using short labels.
  The running time of  our size discovery algorithm is $O(Dn^2\log \Delta)$, for $n$-node networks of diameter $D$ and maximum  degree $\Delta$. We did not
  try to optimize this running time. A natural open question is: what is the fastest size discovery algorithm using a shortest possible labeling scheme, i.e., a scheme of length $\Theta(\log\log \Delta)$?
  On the other hand, the running time of our diameter discovery algorithm is $O(D\log D)$. It is also natural to ask what is the optimal time of a diameter discovery algorithm using a labeling scheme of constant length. Another direction of future research could be considering our problems in the context of dynamic networks and/or in the presence of faults.

%


\end{document}